\def\isarxiv{1}

\def\paperTitle{Tensor Hinted Mv Conjectures}

\def\paperAuthor{
Zhao Song\thanks{\texttt{magic.linuxkde@gmail.com}. The author would like to thank Jan van den Brand for helpful discussions on formulating this conjecture.} % UC Berkeley.}
}

\ifdefined\isarxiv
\documentclass[11pt]{article}
\usepackage[numbers]{natbib}
\else
\documentclass{article} % For LaTeX2e
\usepackage{iclr2026_conference,times}

\fi

\ifdefined\isarxiv

\usepackage{amsmath}
\usepackage{amsthm}
\usepackage{amssymb}
\usepackage{algorithm}
\usepackage{subfig}
\usepackage{algpseudocode}
\usepackage{graphicx}
\usepackage{grffile}
\usepackage{wrapfig,epsfig}
\usepackage{url}
\usepackage{xcolor}
\usepackage{epstopdf}

\usepackage{bbm}
\usepackage{dsfont}

\else 

\usepackage{hyperref}
\usepackage{url}

\fi

\ifdefined\isarxiv

\else
\usepackage{amsmath}
\usepackage{amsthm}
\usepackage{amssymb}
\usepackage{algorithm}
\usepackage{subfig}
\usepackage{algpseudocode}
\usepackage{graphicx}
\usepackage{grffile}
\usepackage{wrapfig,epsfig}
\usepackage{url}
\usepackage{xcolor}
\usepackage{epstopdf}

\usepackage{bbm}
\usepackage{dsfont}

\fi
 
\allowdisplaybreaks

\ifdefined\isarxiv

\usepackage{tikz}
\usepackage{hyperref}  
\hypersetup{colorlinks=true,citecolor=blue,linkcolor=blue} 
\usetikzlibrary{arrows}
\usepackage[margin=1in]{geometry}

\else
 
\fi
 
\graphicspath{{./figs/}}

\theoremstyle{plain}
\newtheorem{theorem}{Theorem}%[section]
\newtheorem{lemma}[theorem]{Lemma}
\newtheorem{definition}[theorem]{Definition}

\newtheorem{conjecture}[theorem]{Conjecture}

\ifdefined\isarxiv

\else
\renewcommand\cite\citep
\fi

\newcommand{\R}{\mathbb{R}}

\begin{document}

\ifdefined\isarxiv
%%% The below part is the title and author of ArXiv version.

\date{}
\title{\paperTitle}
\author{\paperAuthor}

\else

\title{\paperTitle}

\maketitle

\fi

\ifdefined\isarxiv
 
  \maketitle
  \begin{abstract}
    Brand, Nanongkai, and Saranurak \cite{bns19} introduced a conjecture known as the Hinted Mv Conjecture. Although it was originally formulated for the matrix case, we generalize it here to the tensor setting.

  \end{abstract} 

\else

\begin{abstract}
Brand, Nanongkai, and Saranurak \cite{bns19} introduced a conjecture known as the Hinted Mv Conjecture. Although it was originally formulated for the matrix case, we generalize it here to the tensor setting.
\end{abstract}

\fi

%%% The part below is the main body and reference

%%% This file is the structure for main body content
%%% This file should only contain use \input{xxx}.
%%% TeX files for body contents should be named as:
%%% 01_xxxx.tex
%%% 02_xxxx.tex
%%% ...
%%% 49_xxxx.tex

\section{Introduction}

In \cite{bns19}, they proposed the Hinted Mv conjecture. Here we generalize it to the tensor setting.
\begin{definition}[A Tensor Version of Hinted Mv, Type I]\label{def:tensor_hinted_mv}
Consider a boolean semi-ring and a parameter $0 < \tau \le 1$. Let $k$ denote a positive integer. We define the Tensor Hinted Mv problem through these three phases:
\begin{enumerate}
    \item Input $k+1$ $n \times n$ matrices $M$ and $V_1$, $V_2$, $\cdots, V_k$.
    \item Input $k$ $n \times n$ matrices $P_1, P_2, \cdots, P_k$ with at most $n^\tau$ non-zero entries.
    \item Input a single index $i \in [n]$.
    \begin{itemize}
        \item We need to answer $M ( \oslash_{j=1}^k P_j )^\top ( \oslash_{j=1}^k V_j )_{*,i}$, here $V_{*,i} \in \mathbb{R}^n$ is the $i$-th column of matrix $V = \oslash_{j=1}^k V_j$.
    \end{itemize}
\end{enumerate} 
\end{definition}
Given matrices $K_j \in \R^{n_{\ell} \times d}$ for all $\ell \in [k]$, we define the matrix $K := K_1 \oslash K_2 \oslash \cdots \oslash K_k \in \R^{n_1 n_2 \cdots n_k \times d}$ as follows $K_{ i , j } := \prod_{\ell=1}^k (K_\ell)_{i_\ell,j} $,  $\forall i_{\ell} \in [n_{\ell}]$, here $i := \sum_{\ell=1}^{k-1} (i_{\ell} - 1) \prod_{\ell'=\ell+1}^k n_{\ell'} + i_k $. We denote $\oslash_{j=1}^k V_j := V_1 \oslash V_2 \oslash \cdots \oslash V_k$. For $k=1$,  Definition~\ref{def:tensor_hinted_mv} degenerates to \cite{bns19}(The problem is called Mv-Hinted Mv in their original paper). At first glance, the problem appears to require exponential dependence on $k$. However, by employing the tensor-trick, we demonstrate that polynomial dependence is sufficient. The core idea of the running time analysis is a tensor-trick (Lemma~\ref{lem:oslash}), i.e., $( \oslash_{j=1}^k P_j )^\top ( \oslash_{j=1}^k V_j) = \odot_{j=1}^k ( P_j^\top V_j )$. We use $\odot$ to denote Hadamard product, which is entry-wise product $(A\odot B)_{i,j} = A_{i,j} B_{i,j}$. We denote $\odot_{j=1}^k V_j := V_1 \odot V_2 \odot \cdots \odot V_k$.

Assuming polynomial preprocessing time, it remains unclear how to effectively leverage the structural information provided by $M$ and the $\{ V_{\ell} \}_{\ell \in [k]}$ matrices. Furthermore, we cannot pre-determine which specific entries of $M$ will interact with the entries of the final product $V$. Given the constraint of polynomial time, we are unable to exhaustively evaluate the exponentially large space of possible configurations for $P$. Consequently, we consider two distinct algorithmic strategies for the subsequent phases:

{\bf Method 1.} In phase 2, we just compute $M ( \oslash_{j=1}^k P_j )^\top ( \oslash_{j=1}^k V_j) $. Using Lemma~\ref{lem:oslash}, it can be written as $M ( \odot_{j=1}^k ( P_j^\top V_j )  )$.  For each $j \in [k]$, we compute $P_j^\top V_j$ this takes $n^{1+\tau}$ time, and note that $P_j^\top V_j$ has $n^{\tau}$ non-zero rows. For all the $k$ terms, we need to pay $k \times n^{1+\tau}$ time. Computing $\odot_{j=1}^k (P_j^\top V_j) $ takes $k \cdot n^{1+\tau}$ time. The matrix $\odot_{j=1}^k (P_j^\top V_j) $ also only has at most $n^{\tau}$ non-zero rows. Finally we can compute $M ( \odot_{j=1}^k (P_j^\top V_j)  )$, this takes $n^{\omega(1,1,\tau)}$. Overall it takes $O( n^{\omega(1,1,\tau)} + k n^{1+\tau} )$ time. In phase 3, we simply output the column vector that we care about.

{\bf Method 2.}  We do not compute anything in phase 2, but remember the non-zero sets for $P_j$s. In phase 3 we compute $M( \oslash_{j=1}^k P_j )^\top ( \oslash_{j=1}^k V_j)_{*,i}$. For each $j$, we use $v_j$ to denote the $i$-th column of matrix $V_j$. Then the quantity we need to care about is $M( \oslash_{j=1}^k P_j )^\top ( \oslash_{j=1}^k v_j )$. Using Lemma~\ref{lem:oslash}, the target quantity can be written as $M ( \odot_{j=1}^k P_j^\top v_j )$. For each $j \in [k]$, we compute $P_j^\top v_j$ it takes $n^{\tau}$ time. Since there are $k$ terms, it takes $k \cdot n^{\tau}$. Next, we will compute $ \odot_{j=1}^k ( P_j^\top v_j )$, this takes $k \cdot n^{\tau}$ time. Finally computing $M$ with that $n^{\tau}$ sparse vector will take $n^{1+\tau}$ time. Overall, it can be done in $O(n^{1+\tau} + k n^{\tau })$ time.

We state our conjecture as follows:
\begin{conjecture}
For any algorithm solving Tensor Hinted Mv (Definition~\ref{def:tensor_hinted_mv}, Type I), given polynomial preprocessing in Phase 1, the following lower bounds apply to at least one phase:
\begin{itemize}
\item  Phase 2 requires $\Omega(n^{\omega(1,1,\tau)-\delta} + k n^{1+\tau -\delta}  )$.
\item  Phase 3 requires $\Omega(n^{1+\tau-\delta} + k n^{\tau-\delta})$.
\end{itemize}
For every $\delta > 0$.
\end{conjecture}
For the case $k=1$, the above conjecture degenerates to Conjecture 5.7 in \cite{bns19}.
 Here we provide a proof for the tensor-trick being used in this paper.
\begin{lemma}\label{lem:oslash}
For each $\ell \in [k]$, we define $A_{\ell} \in \R^{n_{\ell} \times d_a}$. Let $n:=\prod_{\ell=1}^k n_{\ell}$. Let $A := ( \oslash_{\ell=1}^k A_{\ell} ) \in \R^{ n \times d_a}$. Let $B := ( \oslash_{\ell=1}^k B_{\ell} ) \in \R^{ n \times d_b}$. 
 We define $C \in \R^{d_a \times d_b}$ as $C := A^\top B$.
 We define $ C_{\ell} := A_{\ell}^\top B_{\ell}$. 
Then, we have $\odot_{\ell=1}^k C_{\ell} = C$.

\end{lemma}
\begin{proof}
For each $i_{\ell} \in [n_{\ell}]$, let $a_{\ell,i_{\ell}}^\top$ denote the $i_{\ell}$-th row of $A_\ell$. Similarly, we define $b$ for $B$. Then, we can write $C \in \R^{d_a \times d_b}$ as $C  = A^\top B  =  \sum_{i=1}^{n} A_{i,*}  ( B_{i,*} )^\top  =  \sum_{i_1,i_2,\cdots,i_k}  ( \odot_{\ell=1}^k a_{\ell,i_{\ell}} ) \cdot (\odot_{\ell=1}^k b_{\ell,i_{\ell}}   )^\top
 = \sum_{i_1,i_2,\cdots,i_k} \odot_{\ell=1}^k ( a_{\ell,i_{\ell}} b_{\ell,i_{\ell}}^\top )=\odot_{\ell=1}^k ( \sum_{i_{\ell}=1}^{n_{\ell}} a_{\ell,i_{\ell}} b_{\ell,i_{\ell}}^\top ) = \odot_{\ell=1}^k C_{\ell}$. Thus, we can conclude $ C = \odot_{\ell=1}^k C_{\ell}$.
\end{proof}

In addition to Definition~\ref{def:tensor_hinted_mv}, we consider an alternative tensor formulation.
\begin{definition}[A Tensor Version of Hinted Mv, Type II]\label{def:tensor_hinted_mv:otimes}
Consider a boolean semi-ring and a parameter $\tau > 0$. We define the Tensor Hinted Mv problem through these three phases:
\begin{enumerate}
    \item Input $k$ $n \times d$ matrices $V_1$, $V_2$, $\cdots, V_k$.
    \item Input a $d \times d \times  \cdots \times d$ (diagonal) tensor $P$ with at most $n^\tau$ non-zero entries.
    \item Input an index set $\{ \ell_1, \cdots, \ell_s \} \subseteq [k]$ and an index multi-set $\{ i_1, \cdots, i_s \} $ where each element is from $[n]$.
    \begin{itemize}
        \item We need to answer $[ P(V_1,V_2,\cdots,V_k) ]_{\cdots,i_1,\cdots, i_s, \cdots }$, where $i_1$ is in the $\ell_1$-th direction. This applies analogously to the other $s-1$ pairs.
    \end{itemize}
\end{enumerate} 
\end{definition}

Here $P(V_1, \cdots, V_k)$ is an $n^{k}$ size tensor where $(i_1,\cdots,i_{k})$-th entry is $\sum_{j_1,\cdots, j_k} P_{j_1,\cdots,j_k} (V_1)_{i_1,j_1} \cdots (V_k)_{i_k,j_k}$. By our definition, it is not hard to check that $[ P(V_1,V_2,\cdots,V_k) ]_{\cdots,i,\cdots}$ is $P(V_1, \cdots, ( V_{\ell})_{*,i}, \cdots, V_k )$. For $k=2$,  Definition~\ref{def:tensor_hinted_mv:otimes} degenerates to \cite{bns19}. Now, let us give some straightforward methods.

{\bf Method 1.} In phase 2, we will compute $P(V_1,V_2,\cdots,V_k)$. This is actually the same as $U_1 \otimes U_2 \otimes \cdots \otimes U_k$\footnote{The $(i_1,\cdots,i_k)$-th entry of this tensor is $\sum_{j_1,\cdots,j_k} \prod_{\ell=1}^k U_{i_{\ell},j_{\ell}} $.} where $U_i$ has size $n \times n^{\tau}$. $U_1$ is the matrix generated by selecting $n^{\tau}$ columns in $V_1$ with rescaling of $V_1$ by $P$. Other $U_{\ell}$ are formed simply by selecting $n^{\tau}$ columns from $V_{\ell}$. We first compute $W_1 = \oslash_{\ell=1}^{k-\lfloor k/2\rfloor} U_{\ell}$, and then compute $W_2 = \oslash_{\ell = k-\lfloor k/2\rfloor + 1}^k U_{\ell}$. That step takes time $O( n^{k-\lfloor k/2\rfloor} )$. Next we will compute $W_1 W_2^\top$. This takes $n^{\omega(k-\lfloor k/2\rfloor, \tau, \lfloor k/2\rfloor)}$ time. In Phase 3, we will output the entries directly, since $W_1W_2^\top$ is just the matrix view of the original tensor.

{\bf Method 2.} In phase 2, we will do nothing. In phase 3, we will compute the answer directly. Without loss of generality, we can assume that $\{\ell_1, \cdots, \ell_s\}$ are just the last $s$ indices in $[k]$. The quantity of interest is essentially $U_1 \otimes \cdots \otimes U_{k-s}$ where $U_1$ is the matrix generated by  selecting $n^{\tau}$ columns (according to $P$) from $V_1$ with rescaling by $P$ as well as those $(U_{\ell_t})_{i_t,*}$ for all $t\in [s]$. Other $U_{\ell}$ are formed simply by selecting columns (according to $P$) from $V_{\ell}$ for $\ell \in \{2, \cdots, k-s\}$. We first compute $W_1 = \oslash_{\ell=1}^{k-s-\lfloor (k-s)/2\rfloor} U_{\ell}$, and then compute $W_2 = \oslash_{\ell = k-s-\lfloor (k-s)/2\rfloor + 1}^{k-s} U_{\ell}$. Those two steps take time $O( n^{k-s-\lfloor (k-s)/2\rfloor} )$. Next we will just compute $W_1 W_2^\top$. This takes $n^{\omega(k-s-\lfloor (k-s)/2\rfloor, \tau, \lfloor (k-s)/2\rfloor)}$ time. We will output $W_1W_2^\top$, since it is just the matrix view of the target tensor.

\begin{conjecture}
For any algorithm solving Tensor Hinted Mv (Definition~\ref{def:tensor_hinted_mv:otimes}, Type II), given polynomial preprocessing in Phase 1, the following lower bounds apply to at least one phase:
\begin{itemize}
\item  Phase 2 requires $\Omega(n^{\omega( k -\lfloor k/2 \rfloor , \lfloor k/2 \rfloor, \tau)-\delta}  )$.
\item  Phase 3 requires $\Omega(n^{\omega(k-s- \lfloor (k-s)/2 \rfloor, \lfloor (k-s)/2 \rfloor , \tau)-\delta} )$.
\end{itemize}
For every $\delta > 0$.
\end{conjecture}

For the case $k=2$ and $s=1$, the above conjecture degenerates to Conjecture 5.7 in \cite{bns19}.

\ifdefined\isarxiv
\bibliographystyle{alpha}
\bibliography{ref}
\else
\bibliographystyle{alpha}
\bibliography{ref}
\fi

%%% The part below is the appendix

\newpage
\onecolumn
\appendix

\end{document}